\documentclass[a4paper,12pt,notitlepage,twocolumn]{article}
\usepackage[a4paper, left=1.5cm, right=1.5cm, top=2.5cm, bottom=2.5cm, headsep=1.2cm]{geometry} 
\usepackage[utf8x]{inputenc}
\usepackage[T1]{fontenc}
\usepackage{graphicx}
\usepackage{enumerate}
\usepackage{float}
\usepackage{indentfirst}
\usepackage{amsmath}
\usepackage{amsfonts}
\usepackage{wrapfig} 
\usepackage{subfig}
\usepackage{longtable}
\usepackage{enumitem}
\usepackage{bbm}
\usepackage{amsthm}
\usepackage{multirow}
\usepackage{mathtools}
\usepackage{url}
\usepackage{authblk}

\def\mc {\mathcal}

\def\mc {\mathcal}

\def\ZZ {{\mathbb Z}}
\def\CC {{\mathbb C}}
\def\RR {{\mathbb R}}
\def\bone {\mathbbm{1}}

\newtheorem{theorem}{Theorem}[section]
\newtheorem{lemma}[theorem]{Lemma}
\newtheorem{fact}[theorem]{Fact}

\newtheorem{definition}{Definition}[section]

\newtheorem{example}{Example}[section]

\begin{document}

\title{
Non-abelian anyons on graphs from presentations of graph braid groups
}
\author{Tomasz Maci\k{a}\.{z}ek$^{1,2}$}
\affil{$^{1}$ Center for Theoretical Physics, Polish Academy of Sciences, Al. Lotnik\'ow 32/46, 02-668 Warszawa, Poland \\ $^2$ School of Mathematics, University of Bristol, Bristol BS8 1TW, UK}
\date{}
\twocolumn[
\begin{@twocolumnfalse}
\maketitle
    \begin{abstract}
The aim of this paper is to analyse algorithms for constructing presentations of graph braid groups from the point of view of anyonic quantum statistics on graphs. In the first part of this paper, we provide a comprehensive review of an algorithm for constructing so-called minimal Morse presentations of graph braid groups that relies on discrete Morse theory. Next, we introduce the notion of a physical presentation of a graph braid group as a presentation whose generators have a direct interpretation as particle exchanges. We show how to derive a physical presentation of a graph braid group from its minimal Morse presentation. In the second part of the paper, we study unitary representations of graph braid groups that are constructed from their presentations. We point out that algebraic objects called moduli spaces of flat bundles encode all unitary representations of graph braid groups. For $2$-connected graphs, we conclude the stabilisation of moduli spaces of flat bundles over graph configuration spaces for large numbers of particles. Moreover, we set out a framework for studying locally abelian anyons on graphs whose non-abelian properties are only encoded in non-abelian topological phases assigned to cycles of the considered graph.
    \end{abstract}
  \end{@twocolumnfalse}
  ]

\section{Introduction}
Anyonic quantum statistics is a notion that refers to situations when an interchange of (quasi)partiles in a physical model results with some general unitary transformation of a possibly multicomponent many-body wave function. Quasi-particles that obey anyonic statistics are called anyons. They are generalisations of bosons and fermions in the following sense. If a pair of bosons is exchanged, the many-particle wave function remains unchanged, i.e. is multiplied by the trivial phase factor $e^{i 0}$. On the other hand, an exchange of two fermions results with the multiplication of the wave function by factor $-1=e^{i \pi}$. For single-component wave functions, an exchange of a pair of anyons results with the multiplication by factor $e^{i\theta}$, $\theta\in ]0,\pi[$. Such scalar anyons are known to appear, for instance, in certain ansatzes for multi-electron wave functions realising the Fractional Quantum Hall effect (FQHE) \cite{Laughlin,Wilczek}. More specifically, they approximately describe excited states of FQHE hamiltonians. FQHE hamiltonians also provide models for anyons described by multi-component wave functions, called non-abelian anyons, see e.g. \cite{Haldane}. While there exist physical models realising anyons on graphs, this field of study is still quite unexplored. The already existing models have found use in quantum computing \cite{alicea} and in solid state physics \cite{HKR,Bolte13a}. 

The a priori existence of different types of anyons is strongly restricted by the topology of the space where the anyons are constrained to move. For instance, scalar anyons do not exist in the three-dimensional Euclidean space, $\RR^3$ \cite{LM}. The same holds true when anyons are constrained to move on a closed orientable two-manifold \cite{Sudarshan}. For $\RR^2$ the existence of scalar anyons is allowed and there are no restrictions for the exchange phase $\theta$. If anyons are constrained to move on a sphere, then the allowed exchange phases are $\theta=n\pi/N$, where $N$ is the number of anyons and $0\leq n\leq 2N-3$ \cite{ThoulessWu85}. It is not clear how to realise anyon exchange on the line, $\RR$, as it is not possible there to exchange particles without a collision. However, on graphs, i.e. on networks built out of one-dimensional line segments, the existence of many junctions allows for a well-defined particle exchange without collisions. This fact has been explored in recent papers \cite{HKRS,ASphd,MS16,MS18} to set out a framework for studying abelian and non-abelian anyons on graphs. In particular, it has been shown that different types of quantum statistics are possible on graphs, depending on the topology of a given graph. For scalar anyons, only bosons and fermions are possible on $3$-connected graphs, whereas on $2$-connected and $1$-connected graphs a great variety of abelian anyons is possible \cite{HKRS}. Much less is known about non-abelian anyons on graphs. By computing certain topological invariants of graph configurations spaces called homology groups \cite{MS16,MS18}, using arguments based on K-theory, it has been shown that for wave functions with a sufficiently large number of components, for many families of graphs there is just one class of non-abelian quantum statistics.

In this paper, we focus on modelling non-abelian anyons on graphs via unitary representations of graph braid groups. Let us next briefly revisit main steps of this construction. For $N$ particles constrained to move in a topological space $X$, we consider wave functions as functions from the $n$-particle configuration space, $C_N(X)$, to complex numbers, $\CC$. The considered wave functions can have more than one component. If this is the case, the $k$-component wave function is described by a vector $\overline\Psi=(\Psi_1(q),\hdots,\Psi_k(q))$, where $q$ describes a configuration of $N$ particles in $X$. Configuration space $C_N(X)$ encodes some basic properties of the studied particles. In particular, we consider only hard-core particles, i.e. from the traditional $N$-fold cartesian product, $X^N$, we exclude collision points given by $\Delta=\{(q_1,\hdots,q_n):\ \ \exists_{i\neq j}\ q_i=q_j\}$. Furthermore, we impose the indistinguishability of particles by identifying configurations that differ by a permutation of particles. This can be written concisely as the quotient $C_N(X)=(X^N-\Delta)/S_N$. It is a well-known fact that such configuration spaces lead to a correct description of anyonic quantum statistics \cite{LM,Souriau,Wilczek}. Another crucial ingredient is the notion of a parallel transport of wave functions around loops in $C_N(X)$. If $X$ is a manifold, one defines a quantum theory by considering a vector bundle over $C_N(X)$. Wave functions are interpreted as sections of such a vector bundle and gauge potentials are incorporated as connections on the considered vector bundle. Recall that in such a setting, flat connections correspond to the vanishing of classical forces in the considered quantum system. This happens, for instance, when a screened magnetic field is present in the system so that it vanishes in the region where the particles are allowed to move. However, a magnetic potential can still be present and can affect the behaviour of the quantum system. Such a flat connection leads to the parallel transport, $\hat T$, that for a given loop $\gamma\subset C_N(X)$  i) transforms wave functions via unitary operators $\hat T_\gamma\overline\Psi=U_\gamma\overline\Psi$, $U_\gamma\in U(k)$, ii) operators depend only on the homotopy class of loops, i.e. $U_\gamma=U_{\gamma'}$ if $\gamma$ is homotopy equivalent to $\gamma'$. This gives rise to a unitary representation of the fundamental group of $C_N(X)$ which is called the $n$-strand braid group of $X$ and denoted by $Br_N(X)$. Therefore, in general, different quantisations of a classical system described by configuration space $C_N(X)$ are in a one-to-one correspondence with isomorphism classes of irreducible unitary representations of $Br_N(X)$. A related mathematical object is called the moduli space of flat bundles given by the quotient 
\begin{equation}
\mathcal{M}_N(X,U(k)):=\frac{{\mathrm{Hom}}(Br_N(X),U(k))}{U(k)}
\end{equation}
In other words, all non-abelian quantum statistics for particles constrained to move in topological space $X$ are given by points of $\mathcal{M}_N(X,U(k))$, while scalar quantum statistics correspond to $\mathcal{M}_N(X,U(1))$, i.e. abelian representations of $Br_N(X)$.

In the main body of this paper we review chosen algorithms for constructing presentations of graph braid groups \cite{FSbraid,KoPark}, i.e. groups $Br_N(X)$ where $X=\Gamma$, a graph. The aim of the first part of the paper is to provide a comprehensive overview of an algorithm for constructing so-called minimal presentations of graph braid groups \cite{KoPark}. In the second part we analyse the algorithm from the point of view of anyonic quantum statistics, i.e. unitary representations of graph braid groups constructed from their minimal presentations. In particular, i) we provide arguments for the stabilisation of $\mathcal{M}_N(\Gamma,U(k))$, i.e. if $\Gamma$ is $2$-connected, there exists $N_0$ such that for all $N> N_0$ we have $\mathcal{M}_N(\Gamma,U(k))\cong \mathcal{M}_{N_0}(\Gamma,U(k))$, ii) in analogy to anyons on a torus \cite{Einarsson}, we define so-called locally abelian anyons on graphs which are anyons that locally braid as abelian anyons, but globally behave in a non-abelian way. Throughout the paper, we analyse examples of graphs and their corresponding braid groups and derive their minimal presentations in terms of loops in $C_N(\Gamma)$.

\section{Presentations of graph braid groups - a review}
Graph configuration spaces are aspherical, i.e. their fundamental group is their only non-vanishing homotopy group. Equivalently, the universal covering space of $C_N(\Gamma)$ is contractible. Therefore, graph braid groups encode all topological information about graph configuration spaces (see e.g. \cite{Luck}). However, finding the form of $Br_N(\Gamma)$ for a given graph is known to be a difficult task. All graph braid groups (we restrict our attention only to finite graphs) are finitely presented. This means that there exists a finite set of generators $\alpha_1,\hdots,\alpha_r$ and a finite set of relators $R_1(\alpha_1,\hdots,\alpha_r),\hdots,R_s(\alpha_1,\hdots,\alpha_r)$ in the form of finite words in $\alpha_1,\hdots,\alpha_r$ and their inverses such that 
\begin{gather}\label{presentation}
\nonumber Br_N(\Gamma)=\langle \alpha_1,\hdots,\alpha_r|\ R_1(\alpha_1,\hdots,\alpha_r)=1, \hdots, \\R_s(\alpha_1,\hdots,\alpha_r)=1\rangle.
\end{gather}
Equation (\ref{presentation}) is called a presentation of group $Br_N(\Gamma)$. There exists a certain intuitive choice of generators for $Br_N(\Gamma)$ in terms of particles moving on junctions and loops in $\Gamma$ \cite{HKRS,Knudsen}. However, this intuitive choice of generators leads to many redundancies and can be greatly simplified. Moreover, except for the two-particle case \cite{kurlin}, it is not clear how to complete such a description and write down the set of relators. Therefore, we will shortly proceed with a different method that relies on discrete Morse theory \cite{FSbraid, KoPark} and leads to a minimal presentation of $Br_N(\Gamma)$ as the fundamental group of a much smaller space called the Morse complex of $C_N(\Gamma)$ and denoted here by $\tilde D_N(\Gamma,T)$ where $T$ is a spanning tree of $\Gamma$. One of the drawbacks of the Morse-complex method is that additional work has to be done in order to interpret generators as loops back in $C_N(\Gamma)$. Nevertheless, we show how one can accomplish such an interpretation and we realise it in examples.

Let us start with the aforementioned intuitive set of generators. The construction of such a set relies on an analysis of two small canonical graphs. The first canonical graph is a $Y$-graph which describes an exchange of a pair of particles on a junction in $\Gamma$. The exchange is called a $Y$-exchange and goes as shown in Fig. \ref{y-gen}, left to right.
\begin{figure}
\centering
\includegraphics[width=.45\textwidth]{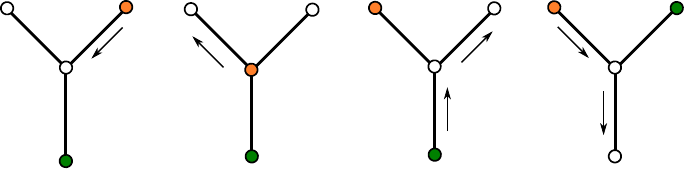}
\caption{A generator of $Br_N(\Gamma)$ as an exchange of a pair of particles in a $Y$-junction. The positions of the remaining $N-2$ particles are fixed.}
\label{y-gen} 
\end{figure}
The second graph is a lasso graph (also called a lollipop graph) that consists of a circle with a lead attached. The corresponding generator is called an $\mc{O}$-generator and is shown in Fig. \ref{ab-gen}.
\begin{figure}
\centering
\includegraphics[width=0.15\textwidth]{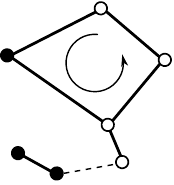}
\caption{A generator of $Br_N(\Gamma)$ where one particle travels around a cycle in $\Gamma$. The positions of the remaining $N-1$ particles are fixed.}
\label{ab-gen} 
\end{figure}
\begin{example}[Two-strand braid group of a $\Theta$-graph]
Group $Br_2(\Gamma_\Theta)$ is a free group on three generators \cite{kurlin}, $Br_2(\Gamma_\Theta)=\langle \alpha_D,\alpha_U,\gamma_L\rangle$. Generators $\alpha_U$ and $\alpha_D$ are of the $\mc{O}$-type while generator $\gamma_L$ denotes a $Y$-exchange on the left junction. Clearly, it is possible to have an analogous exchange on the right junction, $\gamma_R$. Such an exchange depends on the above generators as (see Fig. \ref{theta-proof} for a pictorial proof)
\begin{equation}\label{theta-rel}
\gamma_R\sim\alpha_D\alpha_U\gamma_L^{-1}\alpha_D^{-1}\alpha_U^{-1}.
\end{equation}
 \begin{figure}[H]
\centering
\includegraphics[width=0.45\textwidth]{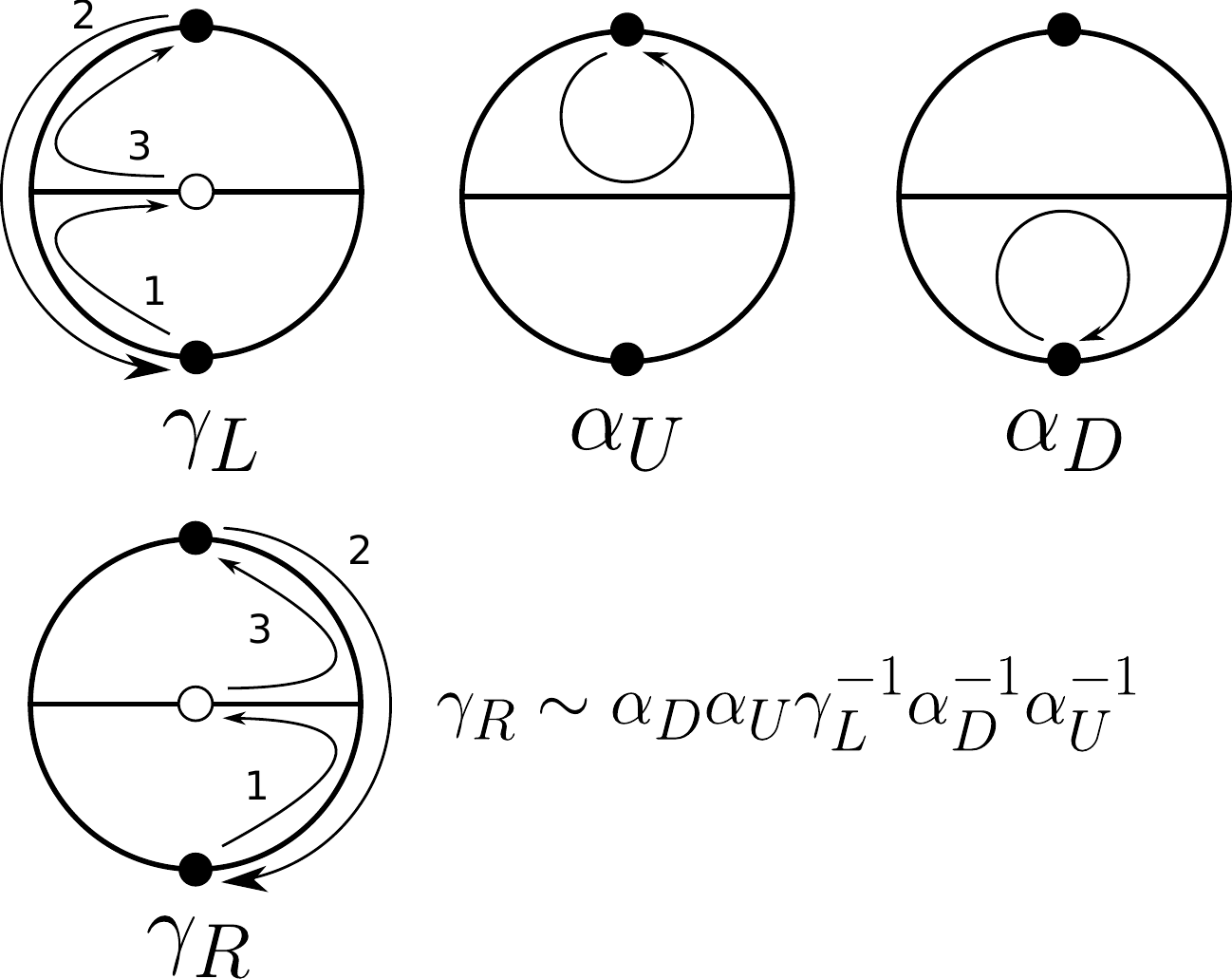}
\caption{Group $Br_2(\Gamma_\Theta)$ is a free group with three generators: $\alpha_U,\alpha_D,\gamma_L$. An exchange on the right junction can be expressed as the above word in the three generators.}
\label{theta-braid}
\end{figure}

 \begin{figure}[H]
\centering
\includegraphics[width=0.45\textwidth]{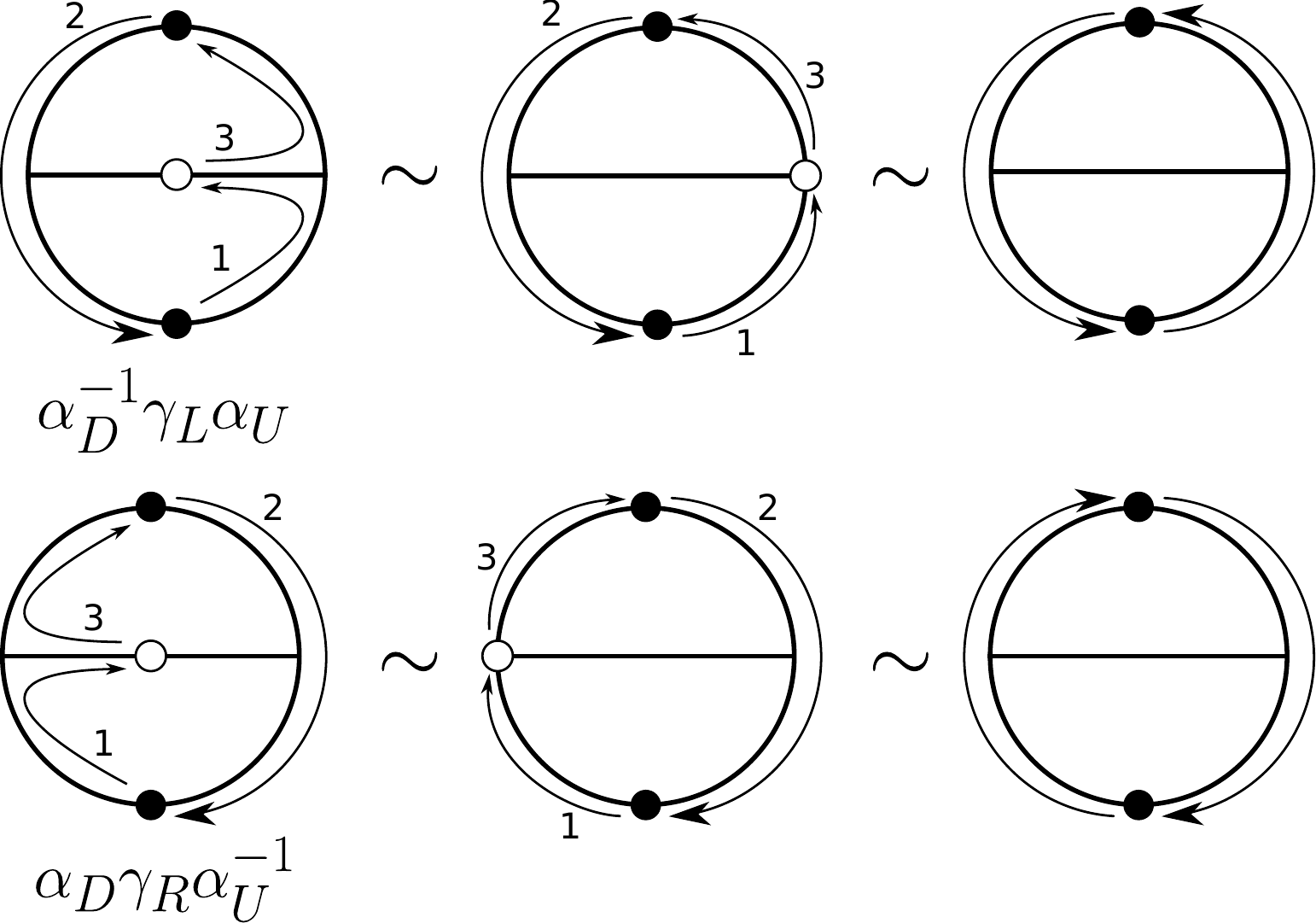}
\caption{A pictorial proof showing that $\alpha_D^{-1}\gamma_L\alpha_U\sim\left(\alpha_D\gamma_R\alpha_U^{-1}\right)^{-1}$, a relation which is equivalent with relation (\ref{theta-rel}) for generators from Fig. \ref{theta-braid}. }
\label{theta-proof}
\end{figure}
\end{example}

\subsection{General properties}
Before we proceed with the discrete Morse theory for graphs, we summarise some general properties of graph braid groups that will play important roles in further sections. Firstly, recall the definition of the commutator subgroup. For any group $G$, its commutator subgroup, denoted here by $G'$, is the group generated by group commutators of elements of $G$
\begin{equation*}
G':=\langle \alpha\beta\alpha^{-1}\beta^{-1}:\ \alpha,\beta\in G\rangle.
\end{equation*}
The quotient $G/G'$ is an abelian group called the abelianisation of $G$. By the asphericity of graph configuration spaces, we have $Br_N(\Gamma)'\cong H_1(C_N(\Gamma),\ZZ)$ where $H_1$ denotes the first homology group. As it has been shown in \cite{HKRS, KoPark}, for any graph we have $H_1(C_N(\Gamma),\ZZ)\cong \ZZ^m\oplus(\ZZ_2)^p$, where exponents $m$ and $p$ depend on $\Gamma$ and $N$. In particular, $p=0$ if and only if $\Gamma$ is planar. For $2$-connected graphs $H_1$ stabilises, i.e. $H_1(C_N(\Gamma),\ZZ)\cong H_1(C_2(\Gamma),\ZZ)$. Recall that $\Gamma$ is $2$-connected if between any two vertices there exist at least two independent paths. Presentation of $Br_N(\Gamma)$ that has $m+p$ generators is called minimal. Such a presentation can be constructed for any graph using Morse-theoretic methods from \cite{KoPark}. Importantly, for planar graphs the existence of a minimal presentation implies that $Br_N(\Gamma)$ is commutator-related (Theorem 4.6 in \cite{KoPark}). This means that relators $\{R_i\}$ in minimal presentations for planar graphs belong to $Br_N(\Gamma)'$.

\subsection{Morse presentations}
Graph configuration spaces have homotopy types of $CW$-complexes. There are different ways to obtain a $CW$-complex as a deformation retract of $C_N(\Gamma)$, one of which is due to Abrams \cite{AbramsPhD} and an other one due to \'{S}wi\k{a}tkowski \cite{swiatkowski}. The algorithm we analyse in this paper relies on Abrams's complex which we denote by $D_N(\Gamma)$. The deformation retraction $C_N(\Gamma)\to D_N(\Gamma)$ is valid if graph $\Gamma$ is {\it sufficiently subdivided}. This means that one has to subdivide edges of $\Gamma$ by adding an appropriate number of vertices of degree $2$ so that the following conditions are met \cite{Abrams-subdiv}. 
\begin{enumerate}
\item Each path between distinct essential vertices (vertices of degree not equal to 2) contains at least $N-1$ edges.
\item Each nontrivial cycle in $\Gamma$ contains at least $N+1$ edges.
\end{enumerate}
An important property of $D_N(\Gamma)$ is that it is a regular cube complex. This means that its cells are cubes that are glued with each other by identifying their faces (gluing maps are injective). Cells of $D_N(\Gamma)$ are denoted as sets of cardinality $N$ whose elements are either edges or vertices of $\Gamma$, all disjoint with each other. While computing graph braid groups we will only be interested in one- and two-dimensional cells of $D_N(\Gamma)$. Hence, let us write down explicitly the general form of a $1$-cell and a $2$-cell. A $1$-cell of $D_N(\Gamma)$ is of the form 
\begin{equation}\label{1cell}
\{e,v_1,\dots,v_{N-1}\}
\end{equation}
where $e\in E(\Gamma)$, $\{v_1,\dots,v_{N-1}\}\subset V(\Gamma)$, $v_i\neq v_j$ for $i\neq j$ and $e\cap v_i=\emptyset$ for all $i$. Similarly, a general $2$-cell of $D_N(\Gamma)$ is of the form 
\begin{equation}\label{2cell}
\{e,e',v_1,\dots,v_{N-2}\}
\end{equation}
where $\{e,e'\}\subset E(\Gamma)$, $\{v_1,\dots,v_{N-2}\}\subset V(\Gamma)$, $v_i\neq v_j$ for $i\neq j$ and $e\cap v_i=e'\cap v_i=\emptyset$ for all $i$. In order to define a boundary map, we choose a spanning tree $T\subset \Gamma$ and order its vertices in the following way. We choose a planar embedding of $T$ and choose a vertex of degree $1$ to be the root of $T$. This choice fully determines a boundary map on $D_N(\Gamma)$, the resulting Morse complex and presentation of $Br_N(\Gamma)$. The root has label $1$. Next, we move along the tree from the root and number the consecutive vertices with consecutive natural numbers. When a junction od degree $d$ is met, the branches are indexed by $0,1,\hdots,d-1$ where branch $0$ is the one that leads to the root and the remaining branches are indexed increasingly in the clockwise direction from branch $0$. The priority in numbering have (unnumbered) vertices that lie on the branch with the lowest index. After finishing the labelling process, the vertices of $\Gamma$ form a totally ordered set. Every edge $e\in E(\Gamma)$ has its initial and final vertex which are denoted by $\iota(e)$ and $\tau(e)$ respectively and satisfy $\tau(e)<\iota(e)$. This gives an orientation of $1$-cells of $D_N(\Gamma)$. Namely, a cell of the form (\ref{1cell}) is oriented from $\{\iota(e),v_1,\dots,v_{N-1}\}$ to $\{\tau(e),v_1,\dots,v_{N-1}\}$. Presentations of $Br_N(\Gamma)$ will be phrased in terms of oriented $1$-cells and their inverses treated as an alphabet. To every $2$-cell (\ref{2cell}) we assign its boundary word as follows (Fig. \ref{fig:2cell})
\begin{figure}
\centering
\includegraphics{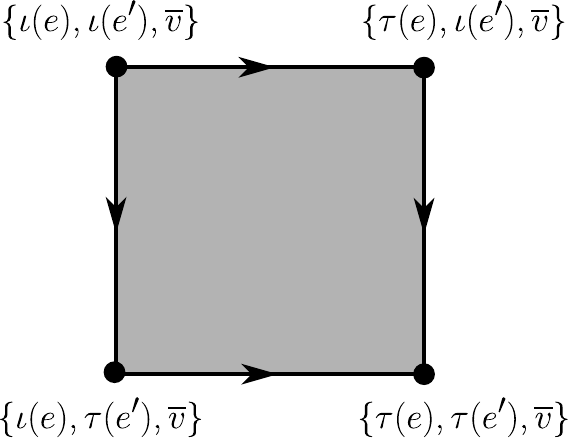}
\caption{A $2$-cell of $D_N(\Gamma)$ and its oriented boundary.}
\label{fig:2cell} 
\end{figure}
\begin{gather}\label{bnd-word}
\nonumber \{e,\iota(e'),\overline{v}\}\{e',\tau(e),\overline{v}\}\{e,\tau(e'),\overline{v}\}^{-1}\times\\\times\{e',\iota(e),\overline{v}\}^{-1}
\end{gather}
where $\overline v$ is a shorthand notation for $v_1,\hdots,v_{N-2}$.

The morse complex $\tilde D_N(\Gamma,T)$ is constructed via a Morse matching $W$ on $D_N(\Gamma)$. $W$ is a collection of  functions $\{W_i\}_{i=0}^{\dim D_N(\Gamma)-1}$, each of which is a function from the set of $i$-cells of $D_N(\Gamma)$ to the set of $i+1$-cells of $D_N(\Gamma)$. Each $W_i$ is a partial function which means that it is not surjective and  its domain is only a subset of $i$-cells, called the set of {\it redundant} $i$-cells. Cells that belong to the image of $W_i$ are called {\it collapsible}. The sets of redundant and collapsible $i$-cells are always disjoint. Moreover, if $W_i(\sigma)=\tau$, then $\tau$ is an $i+1$-cell whose boundary contains cell $\sigma$. Cells which are neither collapsible nor redundant are called {\it critical} and these are the cells that constitute the Morse complex. A Morse matching has to satisfy a few more general conditions, for which we refer the reader to \cite{FS12}. Let us next proceed to the exact form of the Morse matching that we will use. We will focus on functions $W_0$ and $W_1$ as these are the relevant ones in computing $Br_N(\Gamma)$. Intuitively, the Morse matching gives a set of rules to slide particles down the tree $T$ as if the particles were attracted to the root. For any vertex $v\in V(\Gamma)$ we define its corresponding edge $e(v)$ as the unique edge in $T$ which satisfies $\iota(e(v))=v$. Let $\sigma$ be a $0$-cell or a $1$-cell. This means that $\sigma$ is either a subset of $N$ vertices of $\Gamma$ or $\sigma$ is of the form (\ref{1cell}).  We say that vertex $v\in\sigma$ is unblocked if $\left(\sigma-\{v\}\right)\cup e(v)$ is a cell of $D_N(\Gamma)$. In other words, one can slide $v$ down the tree without colliding with other elements of $\sigma$. Otherwise, vertex $v$ is called blocked. Another important notion is the notion of a non-order-respecting edge. An edge $e\in\sigma$ is non-order-respecting if i) $e$ is not in $T$ (in that case $e$ is also called a deleted edge) or ii) there is a vertex $v\in\sigma$ such that $\iota(e)>v>\tau(e)$ and $e(v)\cap e=\tau(e)$. Otherwise, $e$ is order-respecting. Intuitively, this gives a priority rule for particles meeting at junctions of $T$ -- the particle occupying the branch of the lowest index has the priority to move. Critical cells are now easily characterised as those whose {\it all vertices are blocked and all edges are non-order-respecting}. Moreover, we will always choose the spanning tree $T$ so that there is just one critical $0$-cell. Such a critical $0$-cell is necessarily of the form $\{1,2,\hdots,N\}$. It follows that all other $0$-cells of $D_N(\Gamma)$ are redundant. A $0$-cell $\sigma^{(0)}$ is mapped by $W_0$ to a $1$-cell by replacing the lowest unblocked vertex $v\in\sigma$ by its corresponding edge $e(v)$. A $1$-cell, $\sigma^{(1)}$, is redundant if and only if it is not in the image of $W_0$ and it is not critical. If this is the case, then $W(\sigma^{(1)})$ is determined by replacing the lowest unblocked vertex $v\in\sigma^{(1)}$ with $e(v)$. Now we have all the building blocks that are needed to compute the Morse presentation of $Br_N(\Gamma)$.  Denote by $w_i$ an arbitrary word from the alphabet built on oriented $1$-cells of $D_N(\Gamma)$ and their inverses. The following two theorems constitute a foundation of our further considerations.
\begin{theorem}[\cite{FSbraid,FS12}]\label{thm:moves}
$Br_N(\Gamma)$ is generated by all critical $1$-cells subject to relations that come from boundary words (\ref{bnd-word}) of critical $2$-cells by the following set of moves.
\begin{enumerate}
\item Free cancellation. -- If $w=w_1 \sigma \sigma^{-1}w_2$ or $w=w_1 \sigma^{-1}\sigma w_2$, do $w\to w_1w_2$.
\item Collapsing. -- If $w=w_1 \sigma w_2$ or $w=w_1 \sigma^{-1} w_2$ and $\sigma$ is collapsible, do $w\to w_1w_2$.\
\item Simple homotopy. -- If $w=w_1 \sigma w_2$ or $w=w_1 \sigma^{-1} w_2$ and $\sigma w_3$ is a boundary word of a $2$-cell $\tau$ such that $W_1(\sigma)=\tau$, do $w\to w_1w_3^{-1}w_2$ or  $w\to w_1w_3w_2$ respectively.
\end{enumerate}
By iterating the above set of moves, one ends up with an invariant word $\tilde w$ which consists only of critical $1$-cells.
\end{theorem}

\begin{theorem}{\cite{FSbraid}}\label{thm:morse-presentation}
Let $T\subset \Gamma$ be a spanning tree such that the corresponding Morse complex consists of only one critical $0$-cell. Then, 
\begin{equation}\label{morse-presentation}
Br_N(\Gamma)=\left\langle\Sigma^{(1)}|\ \widetilde{b(\tau)}=1, \tau\in\Sigma^{(2)}\right\rangle,
\end{equation}
where $\Sigma^{(i)}$ denotes the set of critical $i$-cells of Morse complex $\tilde D_N(\Gamma,T)$ and $b(\sigma)$ denotes the boundary word of $\sigma$, as in (\ref{bnd-word}).
\end{theorem}
A {\it Python} implementation of the above Theorems \ref{thm:moves} and \ref{thm:morse-presentation} created by the authors of this paper can be found on website \cite{code} in a program which computes Morse presentations of graph braid groups.

\begin{example}[Morse presentations for a $\Theta$-graph for $N\leq 4$.]\label{ex:theta-morse}
\begin{figure}
\centering
\includegraphics[width=.35\textwidth]{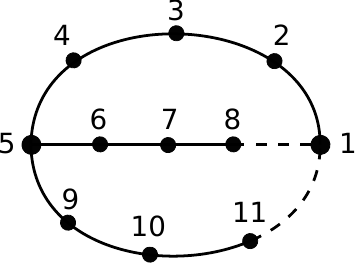}
\caption{A $\Theta$-graph sufficiently subdivided for $N\leq5$ together with a choice of a spanning tree and vertex order.}
\label{fig:theta-morse} 
\end{figure}
Consider a $\Theta$-graph on Fig. \ref{fig:theta-morse} which is sufficiently subdivided for $5$ particles (the reasons for subdividing the graph more than necessary will become clear in Subsection \ref{sec:exchanges}). For $N=2$ we have the following critical $1$-cells in $D_2(\Gamma_{\Theta})$:
\begin{equation}\label{theta-gen-N2}
\tilde\alpha_1=\left\{e_{1}^8,2\right\},\ \tilde\alpha_2=\left\{e_1^{11},2\right\},\ \tilde\gamma=\left\{e_5^9,6\right\}.
\end{equation}
There are no critical $2$-cells in $D_2(\Gamma_{\Theta})$, hence we have reproduced the result from Fig. \ref{theta-braid} -- $Br_2(\Gamma_{\Theta})$ is a free group on $3$ generators (\ref{theta-gen-N2}). For $N=3$, the critical $1$-cells read:
\begin{gather}\label{theta-gen-N3}
\tilde\alpha_1=\left\{e_{1}^8,2,3\right\},\ \tilde\alpha_2=\left\{e_1^{11},2,3\right\},\\ \nonumber \tilde\gamma=\left\{e_5^9,1,6\right\},\\ \nonumber \sigma_1=\left\{e_5^9,6,7\right\},\ \sigma_2=\left\{e_5^9,6,10\right\},
\end{gather}
while the critical $2$-cells read:
\begin{gather*}
\tau_1=\left\{e_{1}^8,e_{5}^9,6\right\},\ \tau_2=\left\{e_{1}^{11},e_{5}^9,6\right\}.
\end{gather*}
It is straightforward to verify (perhaps with the aid of a computer program) that the boundary words are respectively
\begin{gather}\label{theta-rel-N3}
\widetilde{b(\tau_1)}=\tilde\alpha_1\tilde\gamma^{-1} \tilde\alpha_1^{-1}\tilde\gamma^{-1}\sigma_1\\
\nonumber \widetilde{b(\tau_2)}=\tilde\alpha_2\tilde\gamma^{-1} \tilde\alpha_2^{-1}\sigma_2.
\end{gather}
From the corresponding pair of relators we get that i) $\sigma_1=\tilde\gamma\tilde\alpha_1\tilde\gamma\tilde\alpha_1^{-1}$, ii) $\sigma_2=\tilde\alpha_2\tilde\gamma\tilde\alpha_2^{-1}$. Hence, via Tietze transformations we obtain an analogous situation as for $N=2$, i.e.
\[Br_3(\Gamma_\Theta)=\langle\tilde\alpha_1,\tilde\alpha_2,\tilde\gamma\rangle.\]
Finally, let us demonstrate that $Br_4(\Gamma_\Theta)$ is no longer a free group. The critical $1$-cells read:
\begin{gather}\label{theta-gen-N4}
\tilde\alpha_1=\left\{e_{1}^8,2,3,4\right\},\ \tilde\alpha_2=\left\{e_1^{11},2,3,4\right\},\\ \nonumber \tilde\gamma=\left\{e_5^9,1,2,6\right\},\\ \nonumber \sigma_1=\left\{e_5^9,1,6,7\right\},\ \sigma_2=\left\{e_5^9,1,6,10\right\}, \\
\nonumber \sigma_3=\left\{e_5^9,6,7,8\right\},\ \sigma_4=\left\{e_5^9,6,7,10\right\}, \\
\nonumber \sigma_5=\left\{e_5^9,6,10,11\right\}.
\end{gather}
while the critical $2$-cells read:
\begin{gather*}
\tau_1=\left\{e_{1}^8,e_{5}^9,2,6\right\},\ \tau_2=\left\{e_{1}^{11},e_{5}^9,2,6\right\}, \\
\tau_3=\left\{e_{1}^{11},e_{5}^9,6,7\right\},\ \tau_4=\left\{e_{1}^8,e_{5}^9,6,7\right\}, \\
\tau_5=\left\{e_{1}^{11},e_{5}^9,6,10\right\},\ \tau_6=\left\{e_{1}^8,e_{5}^9,6,10\right\}.
\end{gather*}
Boundary words for cells $\tau_1$ and $\tau_2$ are exactly the same expressions as in (\ref{theta-rel-N3}). Besides that, we have
\begin{gather}\label{theta-rel-N4}
\widetilde{b(\tau_3)}=\tilde\alpha_2 \sigma_1^{-1} \tilde\alpha_2^{-1}\sigma_4,\\
\nonumber \widetilde{b(\tau_4)}=\tilde\alpha_1 \sigma_1^{-1}\tilde\alpha_1^{-1}\tilde\gamma^{-1}\sigma_3,\\
\nonumber \widetilde{b(\tau_5)}=\tilde\alpha_2\sigma_2^{-1}\tilde\alpha_2^{-1}\sigma_5, \\
\nonumber \widetilde{b(\tau_6)}=\tilde\gamma\tilde\alpha_1\sigma_2^{-1}\tilde\alpha_1^{-1}\tilde\gamma^{-1}\sigma_2^{-1}\sigma_4.
\end{gather}
To obtain a minimal presentation of $Br_4(\Gamma_\Theta)$ we realise the following Tietze transformations. From $\widetilde{b(\tau_3)}=1$ and from the expression for $\sigma_1$ we extract $\sigma_4=\tilde\alpha_2\tilde\gamma\tilde\alpha_1\tilde\gamma\tilde\alpha_1^{-1}\alpha_2^{-1}$. Similarly, from $\widetilde{b(\tau_4)}=1$ and $\widetilde{b(\tau_5)}=1$ we obtain expressions for $\sigma_3$ and $\sigma_5$ respectively. Hence, the only nontrivial relator in $Br_4(\Gamma_\Theta)$ comes from $\widetilde{b(\tau_6)}=1$ after plugging in expressions for $\sigma_4$ and $\sigma_2$. One can rewrite the result as follows
\begin{equation}\label{minimal-thetaN4}
Br_4(\Gamma_\Theta)=\left\langle\tilde\alpha_1,\tilde\alpha_2,\tilde\gamma|\ \left[\tilde\gamma,Ad_{\tilde\alpha_1\tilde\alpha_2}(\tilde\gamma)\right]=1\right\rangle,
\end{equation}
where we use a shorthand notation $Ad_{h}(g):=h g h^{-1}$.
\end{example}

\subsection{Minimal presentations}\label{sec:minimal}
Exemple \ref{ex:theta-morse} shows some of the crucial features of computations related to Morse presentations of graph braid groups. First of all, the number of generators can be greatly reduced via Tietze transformations by utilising some of the relators. As shown in \cite{KoPark}, this can be done in a systematic way by dividing the set critical $1$-cells into sets of so-called {\it pivotal}, {\it separating} and {\it free} cells. Free cells automatically contribute to the minimal Morse presentation. All pivotal cells and some of the separating cells can be removed via boundary words of appropriate critical $2$-cells. In this section, we will briefly review this construction. Secondly, for a graph which is sufficiently subdivided for $N$ particles, boundary words of critical $2$-cells for $N'$ particles, $N'< N$, are inherited as boundary words of appropriate critical $2$-cells for $N'+1$ particles. We will utilise this fact in section \ref{sec:stabilisation}. We start this section with recalling the following crucial lemma.
\begin{lemma}[Minimal presentations \cite{KoPark}]\label{lemma:minimal}
Group $Br_N(\Gamma)$ has a minimal presentation over $m_{N,\Gamma}+p_{N,\Gamma}$ generators for $m_{N,\Gamma}$ and $p_{N,\Gamma}$ that are natural numbers which determine $H_1(C_N(\Gamma),\ZZ)=\ZZ^{m_{N,\Gamma}}\oplus\left(\ZZ_2\right)^{p_{N,\Gamma}}$. 
\end{lemma}
As a corollary, we obtain that for a planar graph $Br_N(\Gamma)$ has a presentation over $m_{N,\Gamma}$ generators. Moreover, if $\Gamma$ is $2$-connected, then the number of generators of a minimal presentation stabilises with $N$ for $N\geq 2$, i.e. $m_{N,\Gamma}=m_{2,\Gamma}$ and $p_{N,\Gamma}=p_{2,\Gamma}$. This can be observed in example \ref{ex:theta-morse} where for $N\geq 2$ we have $H_1(C_N(\Gamma_\Theta))=H_1(C_2(\Gamma_\Theta))=\ZZ^3$.

In order to find a minimal Morse presentation of $Br_N(\Gamma)$ we have to introduce a few technical notions from paper \cite{KoPark}. However, in order to keep the presentation clear and concise, when possible, we will skip some of the details. 

We say that an edge $e\in E(\Gamma)$ is separated in $T\subset \Gamma$ by $v\in V(\Gamma)$ iff $\iota(e)$ and $\tau(e)$ lie in two distinct connected components of $T-\{v\}$. The first technical step is to choose a spanning tree $T\subset \Gamma$ which satisfies the following conditions of lemma 2.5 in \cite{KoPark}: T1) For every edge $e\in E(\Gamma)-E(T)$ we have that $\iota(e)$ is of valency $2$. T2) Every edge $e\in E(\Gamma)-E(T)$ is not separated in $T$ by any vertex $v\in V(\Gamma)$ such that $v<\tau(e)$. For the sake of completeness, we mention that there is an additional property T3 which is phrased in terms of other geometric properties of $\Gamma$, however we will not write it down here. We only point out that in paper \cite{KoPark} there is an algorithmic way to choose a tree which satisfies properties T1, T2 and T3. The choice of such a tree is essential for definitions of pivotal, separating and free cells to work. One of the key notions is the {\it size} of a critical $1$-cell denoted by $s(\sigma)$. For a critical cell (\ref{1cell}), $s(\sigma)$ is the number of vertices in $\sigma$ that are blocked behind $\tau(e)$ on branches incident to $\tau(e)$ with index greater than $0$. In example \ref{ex:theta-morse}, cells from equation (\ref{theta-gen-N4}) have sizes $s(\tilde\alpha_1)=s(\tilde\alpha_2)=0$, $s(\tilde\gamma)=1$, $s(\sigma_2)=s(\sigma_2)=2$ and $s(\sigma_3)=3$.

In the remaining part of this subsection, we specify our considerations to $2$-connected graphs, as we anticipate that such graphs appear in most of the physically relevant situations. The notion of the size of a critical cell was necessary for introducing a simple criterion for separating out most of the pivotal cells. We state this criterion without a proof in the form of the following fact.
\begin{fact}[\cite{KoPark}]
Every critical $1$-cell $\sigma$ with $s(\sigma)\geq 2$ is pivotal, hence can be expressed as a word in free and separating $1$-cells. 
\end{fact}
It follows that effectively all relevant generators in a minimal Morse presentation of $Br_N(\Gamma)$ appear already on the level of $N=2$. This can be seen by noting that vertices in a critical cell that are blocked behind the root of $T$ can be ignored to give the corresponding critical cell in $D_2(\Gamma)$. In this way, the minimal set of generators of $Br_N(\Gamma)$ can be found only by considering the two-particle case. For $N>2$, additional work has to be done to eliminate new pivotal cells and make appropriate Tietze transformations in order to recover new relators between the minimal generators from the boundary words of critical $2$-cells. This can be done in an algorithmic way by ordering the pivotal $1$-cells and critical $2$-cells in an appropriate way, as decribed in \cite{KoPark}. We anticipate to incorporate this algorithm in our Python implementation \cite{code}.

\subsection{Relating minimal presentations to particle exchanges}\label{sec:exchanges}
As our considerations from the preceding sections show, it is not clear how to connect generators of $Br_N(\Gamma)$ in its Morse presentation with some physical particle exchanges on $\Gamma$. In this subsection we show how this can be accomplished. Presentations of $Br_N(\Gamma)$, where generators can be directly interpreted as particle exchanges will be called {\it physical presentations}. It turns out that physical presentations can be derived from minimal Morse presentations. However, in order to recover particle exchanges from a minimal Morse presentation of $Br_N(\Gamma)$, one usually has to add some new generators and new relators. 

We start by introducing two classes of loops in $D_n(\Gamma)$. Assume that $T$ is a spanning tree of $\Gamma$ which satisfies conditions T1, T2 and T3 described in subsection \ref{sec:minimal}. The first loop is associated with an exchange of a pair of particles on a $Y$-junction in $T$. More precisely, choose a $Y$-subgraph of $T$ which is spanned on vertices $k,l,m,n$ such that $k<l<m<n$ and vertex $l$ has degree at least $3$. To such a $Y$-subgraph we associate the following word which we call the $Y$-loop.
\begin{gather}\label{y-loop}
\gamma_{k,m,n}(\overline v):=\{e_l^n,k,\overline v\}\{e_l^m,k,\overline v\}^{-1}\{e_k^l,m,\overline v\}^{-1}\times \\ \nonumber  \times\{e_l^n,m,\overline v\}^{-1}\{e_l^m,n,\overline v\}\{e_k^l,n,\overline v\}.
\end{gather}
In the above expression, by $\overline v$ we denote a set of $N-2$ vertices of $V(\Gamma)$ such that $\overline v\cap \{k,l,m,n\}=\emptyset$. A key observation is that if all vertices in $\overline v$ ale blocked in cell $\{e_l^n,m,\overline v\}$, then this cell is critical. Furthermore, we have the following lemma.
\begin{lemma}\label{lemma:y-loop}
Let $\gamma$ be a $Y$-loop in $D_N(\Gamma)$ as in (\ref{y-loop}). If $\overline v=\{1,2,\hdots,N-2\}$, then $\gamma$ is mapped to the Morse complex as critical cell $\{e_l^n,m,\overline v\}^{-1}$.
\end{lemma}
\begin{proof}[Sketch of a proof] By the assumption about the form of $\overline v$, cells $\{e_k^l,m,\overline v\},\ \{e_l^m,n,\overline v\},\ \{e_k^l,n,\overline v\}$ are collapsible. Cells $\{e_l^n,k,\overline v\},\ \{e_l^m,k,\overline v\}$ are redundant. To find the image of the redundant cells under the Morse flow we use lemma 2.3 in \cite{KoPark} which shows that they are carried by the Morse flow to collapsible cells $\{e_l^n,1,2,\hdots,N-1\}$ and $\{e_l^m,1,2,\hdots,N-1\}$ respecively.
\end{proof}

The other type of generators are loops associated to oriented simple cycles in $\Gamma$. Such generators will be called $\mc{O}$-loops. Denote by $\mc{O}=v_1\to v_2\to\hdots\to v_p\to v_1$ an oriented simple cycle in $\Gamma$ that passes through the sequence of vertices $(v_1,v_2,\hdots,v_p,v_1)$ where $v_i$ is adjacent in $\Gamma$ to $v_{i-1}$ and $v_{i+1}$ for $i\in \{2,\hdots,p\}$. For any $\overline v$, a set of $N-1$ vertices of $\Gamma$ such that $\mc{O}\cap \overline v=\emptyset$ we define the corresponding $\mc{O}$-loop as the product
\begin{equation}\label{o-loop}
\alpha_{\mc{O}}(\overline v):=\prod_{e\in E(\Gamma)\cap \mc{O}}\{e,\overline v\}^{a_e},
\end{equation}
where $a_e=1$ if the orientation of $e$ inherited from the order of vertices in the spanning tree agrees with the orientation of cycle $\mc{O}$ and $a_e=-1$ otherwise.
\begin{lemma}\label{lemma:o-loop}
Let $\alpha_{\mc{O}}(\overline v)$ be an $\mc{O}$-loop in $D_N(\Gamma)$ as defined in (\ref{o-loop}). Let $\overline v=\{1,2,\hdots,N-1\}$ if for all deleted edges $e\in\mc{O}\cap(E(\Gamma)-E(T))$ we have $\tau(e)>1$ and let $\overline v=\{2,3,\hdots,N\}$ otherwise. Then, word $\alpha_{\mc{O}}(\overline v)$ is mapped to the Morse complex as 
\begin{equation*}
\alpha_{\mc{O}}(\overline v)\mapsto \prod_{e\in \mc{O}\cap(E(\Gamma)-E(T))}\{e,\overline v_e\}^{a_e},
\end{equation*}
where $\overline v_e=\{1,2,\hdots,N-1\}$ if $\tau(e)>1$ and $\overline v=\{2,3,\hdots,N\}$ otherwise.
\end{lemma}
\begin{proof}[Sketch of a proof] 
If $e\in E(T)$, then cell $\{e,\overline v\}$ is collapsible. Otherwise, if $e\in (E(\Gamma)-E(T))$, the image of cell $\{e,\overline v\}$ in the Morse complex can be easily found using lemma 2.3 in \cite{KoPark}.
\end{proof}

The general strategy is to express generators of a minimal Morse presentation of $Br_N(\Gamma)$ as words in $Y$- and $\mc{O}$-loops. There is one technical detail to make sure that all loops are based at the same point given by configuration $\{1,\hdots,N\}$. This can be easily dealt with by conjugating $Y$- and $\mc{O}$-loops with words that connect their initial configurations with the base point. The following lemma allows us to make sure that such a conjugation does not affect the image of $Y$- and $\mc{O}$-loops in the Morse complex. 

\begin{lemma}[\cite{FSbraid}]
The set of collapsible $1$-cells in $D_N(\Gamma)$ is a spanning tree of the $1$-skeleton of $D_N(\Gamma)$. Hence, there exists a path $P_{\overline v}$ from $\{1,\hdots,N\}$ to any configuration $\overline v=\{v_1,\hdots,v_N\}$ that is a word consisting of only collapsible cells.
\end{lemma}

The next crucial step is to find the actual images of $Y$- and $\mc{O}$-loops in the Morse complex. Although the above lemmas \ref{lemma:o-loop} and \ref{lemma:y-loop} provide some simplification, for arbitrary configurations of free particles $\overline v$ this is usually a complicated task. 

\begin{example}[Physical presentations of $Br_N(\Gamma_\Theta)$]\label{ex:theta-physical}
Let us start with $N=2$ and minimal Morse presentation of $Br_2(\Gamma_\Theta)$ given in equation (\ref{theta-gen-N2}). Consider $Y$-loop 
\begin{gather*}
\gamma_{4,6,9}=\{e_5^9,4\}\{e_5^6,4\}^{-1}\{e_4^5,6\}^{-1}\{e_5^9,6\}^{-1}\times \\ \times \{e_5^6,9\}\{e_4^5,9\}.
\end{gather*}
By lemma \ref{lemma:y-loop} we have $\gamma_{4,6,9}\mapsto \{e_5^9,6\}^{-1}=\tilde\gamma^{-1}$. Next, let us take $\mc{O}$-loops $\alpha_{\mc{D}}(\{2\})$ and $\alpha_{\mc{U}}(\{9\})$ where the corresponding simple cycles read
\begin{gather*}
\mc{D}=5\to6\to7\to8\to1\to11\to10\to9\to5,\\
\mc{U}=1\to2\to3\to4\to5\to6\to7\to8\to1.
\end{gather*}
By lemma  \ref{lemma:o-loop} we have $\alpha_{\mc{D}}(\{2\})\mapsto \{e_1^8,2\}\{e_1^{11},2\}^{-1}=\tilde\alpha_1\tilde\alpha_2^{-1}$ and by a direct calculation we find out that $\alpha_{\mc{U}}(\{9\})\mapsto\{e_5^9,6\}\{e_1^8,2\}=\tilde\gamma\tilde\alpha_1$. Hence, we invert the above expressions as
\begin{gather}\label{theta-physical-N2}
\tilde\gamma=\gamma^{-1}, \\
\nonumber \tilde\alpha_1=\gamma\alpha_{\mc{U}}, \\
\nonumber \tilde\alpha_2=\alpha_{\mc{D}}^{-1}\gamma\alpha_{\mc{U}},
\end{gather}
where we denote the $Y$-loop $\gamma_{4,6,9}$ as $\gamma$ and the $\mc{O}$-loops shortly as $\alpha_{\mc{U}}$ and $\alpha_{\mc{D}}$. Because $Br_2(\Gamma_\Theta)$ is free, we simply have 
\[Br_2(\Gamma_\Theta)=\langle \gamma_{4,6,9},\alpha_{\mc{D}}(\{2\}),\alpha_{\mc{U}}(\{9\})\rangle.\]
To rederive relation (\ref{theta-rel}) note first that we can identify $\gamma\equiv \gamma_L$ from Fig. \ref{theta-braid}. Word associated to loop $\gamma_R$ reads
\begin{gather*}
\gamma_R=\{e_1^{11},2\}\{e_1^8,2\}^{-1}\{e_1^2,8\}\{e_1^{11},8\}^{-1}\times \\ \times\{e_1^8,11\}\{e_1^2,11\}^{-1}.
\end{gather*}
By a direct calculation we check that $\gamma_R\mapsto \tilde\alpha_2\tilde\alpha_1^{-1}\tilde\alpha_2^{-1}\tilde\gamma\tilde\alpha_1$ which after substituting expressions (\ref{theta-physical-N2}) yields relation (\ref{theta-rel}) between physical loops.

Let us next immediately skip to $N=4$. We propose a similar set of generators as $\gamma:=\gamma_{4,6,9}(\{1,2\})$, $\alpha_{\mc{D}}:=\alpha_{\mc{D}}(\{2,3,4\})$ and $\alpha_{\mc{U}}:=\alpha_{\mc{U}}(\{9,10,11\})$. Again, by lemmas \ref{lemma:y-loop} and \ref{lemma:o-loop} we obtain that $\gamma\mapsto\tilde\gamma^{-1}$ and $\alpha_{\mc{D}}\mapsto\tilde\alpha_1\tilde\alpha_2^{-1}$. However, for $\alpha_{\mc{U}}$ we have $\alpha_{\mc{U}}\mapsto \tilde\alpha_1^{-1}\tilde\gamma^{-1}\sigma_2^{-1}\sigma_5^{-1}$. This brings new generators, $\sigma_2$ and $\sigma_5$ into play. We would like to replace them with $Y$-loops $\gamma':=\gamma_{4,6,9}(\{1,10\})$ and $\gamma'':=\gamma_{4,6,9}(\{10,11\})$. By a direct computation we check that indeed $\gamma'\mapsto\sigma_2^{-1}$ and $\gamma''\mapsto\sigma_5^{-1}$. At this point we have enough loops to invert the above relations. The result reads
\begin{gather}\label{theta-physical-gen}
\gamma=\tilde\gamma^{-1}, \sigma_2=(\gamma')^{-1}, \sigma_5=(\gamma'')^{-1}\\
\nonumber \tilde\alpha_1=\gamma\gamma'\gamma''\alpha_{\mc{U}}^{-1},\  \tilde\alpha_2=\alpha_{\mc{D}}^{-1}\gamma\gamma'\gamma''\alpha_{\mc{U}}^{-1}.
\end{gather}
As a final step, we rephrase the relator of minimal Morse presentation (\ref{minimal-thetaN4}) in terms of new generators. Moreover, we have to add two new relators that express the dependency of $\gamma'$ and $\gamma''$ on other generators. After a straightforward substitution, the relator from presentation (\ref{minimal-thetaN4}) now reads
\[R_1=\left[\gamma^{-1},Ad_{\gamma\gamma'\gamma''\alpha_{\mc{U}}^{-1}\alpha_{\mc{D}}^{-1}\gamma\gamma'\gamma''\alpha_{\mc{U}}^{-1}}\left(\gamma^{-1}\right)\right].\]
The additional relators are obtained from boundary words (\ref{theta-rel-N4}). In particular, we  have $1=\tilde\alpha_2\tilde\gamma^{-1} \tilde\alpha_2^{-1}\sigma_2$ and $1=\tilde\alpha_2\sigma_2^{-1}\tilde\alpha_2^{-1}\sigma_5$. After substituting Morse generators with expressions (\ref{theta-physical-gen}), we get
\begin{gather}\label{relators-theta-physical}
R_2=Ad_{\alpha_{\mc{D}}^{-1}\gamma\gamma'\gamma''\alpha_{\mc{U}}^{-1}}\left(\gamma\right)\left(\gamma'\right)^{-1}, \\
\nonumber R_3=Ad_{\alpha_{\mc{D}}^{-1}\gamma\gamma'\gamma''\alpha_{\mc{U}}^{-1}}\left(\gamma'\right)\left(\gamma''\right)^{-1}.
\end{gather}
Summing up, we have replaced minimal Morse presentation (\ref{minimal-thetaN4}) with $3$ generators and $1$ relator with a physical presentation with $5$ generating loops
\begin{gather}\label{theta-phys-N4}
Br_4(\Gamma_\Theta)=\langle \gamma,\gamma',\gamma'',\alpha_{\mc{U}},\alpha_{\mc{D}}|\ R_1=1,\\ 
\nonumber R_2=1,\ R_3=1\rangle.
\end{gather}
\end{example}
The above example presents the full complexity of the problem of constructing physical presentations of graph braid groups. A systematic way of constructing such presentations can be summarised in the following points.
\begin{enumerate}
\item Find a minimal Morse presentation of $Br_N(\Gamma)$.
\item Find $Y$- and $\mc{O}$-loops whose images in the Morse complex contain generators of the minimal Morse presentation.
\item If the images of $Y$- and $\mc{O}$-loops from the previous point contain critical cells other than the minimal generators, add $Y$- and $\mc{O}$-loops that map to the new critical cells. Repeat the procedure until a closed system of equations is obtained.
\item Invert the equations to express critical cells as words in $Y$- and $\mc{O}$-loops.
\item Substitute the minimal generators with their corresponding words in $Y$- and $\mc{O}$-loops to rewrite relators of the minimal Morse presentation in terms of words in loops.
\item From boundary words of critical $2$-cells construct new relators that express the dependency of critical cells on the minimal generators. Rewrite the relators in terms of $Y$- and $\mc{O}$-loops.
\end{enumerate}

\section{Stabilisation of $\mathcal{M}_N(\Gamma,U(k))$}\label{sec:stabilisation}
Let us revisit equation (\ref{presentation}). In order to construct a $U(k)$ representation of group $Br_N(\Gamma)$, to each generator we assign a unitary matrix $\alpha_i\mapsto U_i,\ i=1,\hdots,r$. Relators $\{R_i\}_{i=1}^s$ impose polynomial equations for the chosen set of matrices. Hence, we immediately see that $\mathcal{M}_N(\Gamma,U(k))$ is an algebraic variety, i.e. is defined as the zero set of a system of polynomial equations. More precisely, we have 
\begin{equation}
\mathcal{M}_N(\Gamma,U(k))=\mu_N^{-1}(\bone,\hdots,\bone)/U(k),
\end{equation}
where map $\mu_N:\ U(k)^{r}\to U(k)^s$ acts as
\begin{equation*}
\mu_N(U_1,\hdots,U_r)=\left(R_i(U_1,\hdots,U_r)\right)_{i=1}^s.
\end{equation*}
The essential part in establishing stabilisation of $\mathcal{M}_N(\Gamma,U(k))$ is to define a map which allows us to rewrite generators and relators of a minimal presentation of $Br_N(\Gamma)$ as generators and relators of $Br_{N+1}(\Gamma)$. 
\begin{definition}\label{definition:plus}
Assume that $\Gamma$ is sufficiently subdivided for some, possibly large, $N$. For a critical $1$-cell $\sigma\in D_{N'}(\Gamma)$ for $N'<N$ define $\sigma_+$ as $\sigma\cup\{v\}$ for $v$ such that $v$ is the minimal vertex among $\{1,\hdots,N\}$ for which $\sigma\cup\{v\}$ is a critical $1$-cell in $D_{N'+1}(\Gamma)$. Similarly, for $\tau$ a critical $2$-cell define $\tau_+$ as $\tau\cup\{v\}$ for $v$ such that $v$ is the minimal vertex among $\{1,\hdots,N\}$ for which $\sigma\cup\{v\}$ is a critical $1$-cell in $D_{N'+1}(\Gamma)$. We extend map $+$ to words by acting on consequent cells.
\end{definition}

\begin{lemma}
If $\tilde\sigma_1^{a_1}\hdots\tilde\sigma_k^{a_k}$ is the boundary word for a critical $2$-cell $\tau$, i.e. $\widetilde{b(\tau)}=\tilde\sigma_1^{a_1}\hdots\tilde\sigma_k^{a_k}$, then $\widetilde{b(\tau_+)}=(\tilde\sigma_1)_+^{a_1}\hdots(\tilde\sigma_k)_+^{a_k}$.
\end{lemma}
\begin{proof}
Boundary word for $b(\tau_+)$ for $\tau_+=\tau\cup\{v\}$ is obtained from $b(\tau)$ simply by adding vertex $v$ to each cell in $b(\tau)$. Furthermore, if $\sigma\mapsto \tilde\sigma$ then if $v$ is such that $\tilde\sigma_+=\tilde\sigma\cup\{v\}$ then we have $\sigma\cup\{v\}\mapsto \tilde\sigma_+$ under the Morse flow.
\end{proof}

The above lemma directly implies that for $2$-connected graphs generators of the minimal Morse presentation of $Br_{N'}(\Gamma)$, $N'<N$, for a choice of spanning tree $T\subset \Gamma$ are in a one-to-one correspondence with generators of the minimal Morse presentation of $Br_{N'+1}(\Gamma)$ via map $+$ as defined in \ref{definition:plus}. Furthermore, if $R$ is a relator for the above minimal Morse presentation of $Br_{N'}(\Gamma)$, then $R_+$ is a relator for $Br_{N'+1}(\Gamma)$. This means that $\mathcal{M}_{N'+1}(\Gamma,U(k))\subset \mathcal{M}_{N'}(\Gamma,U(k))$ as an algebraic subvariety. In other words, $\mathcal{M}_{N'+1}(\Gamma,U(k))$ satisfies all polynomial equations that define $\mathcal{M}_{N'}(\Gamma,U(k))$ and some additional polynomial equations coming from new relators. Because the number of complex variables is fixed by the number of generators of the minimal Morse presentation and by number $k$, the procedure of adding new equations has to stabilise at some point.

\begin{example}[Space $\mathcal{M}_{4}(\Gamma_\Theta,U(k))$]
Assign $(\tilde\gamma,\tilde\alpha_1,\tilde\alpha_2)\mapsto (U_{\tilde\gamma},U_1,U_2)\subset U(k)^3$. On the level of matrices, relator from presentation (\ref{minimal-thetaN4}) can be rewritten as $[U_{\tilde\gamma},Ad_{U_1U_2}(U_{\tilde\gamma})]=0$, where by square brackets we mean here the algebraic commutator $[A,B]=AB-BA$. Using the conjugation freedom, one can diagonalise both $U_{\tilde\gamma}$ and $Ad_{U_1U_2}(U_{\tilde\gamma})$ at the same time. Note that matrices $U_{\tilde\gamma}$ and $Ad_{U_1U_2}(U_{\tilde\gamma})$ are isospectral. Hence after the aforementioned diagonalisation, conjugation $Ad_{U_1U_2}(U_{\tilde\gamma})$ can only permute eigenvalues of $U_{\tilde\gamma}$. If the spectrum of $U_{\tilde\gamma}$ is non-degenerate, this means that $U_1U_2=e^{i\alpha}P$ where $P$ is a permutation matrix. In other words, $\mathcal{M}_{4}(\Gamma_\Theta,U(k))$ contains $k!$ isotypical connected components $\mathcal{M}_P$ labelled by elements of the symmetric group, $P\in S_k$. Each component is of the form 
\[\mathcal{M}_P\cong U(k)\times U(1)\times \frac{(U(1)^{k}-\Delta)}{S_k}.\] 
Factor $(U(1)^{k}-\Delta)/S_k$ where $\Delta:=\{(z_1,\hdots,z_k)\in U(1)^{k}:\ z_i=z_j\ \mathrm{for\ some}\ i\neq j\}$ corresponds to the quotient of the set of diagonal $U(k)$ matrices with non-degenerate spectra by the action of the Weyl group which permutes the eigenvalues.  A tuple $(U,e^{i\alpha},[(e^{i\phi_1},\hdots,e^{i\phi_k})])\in \mathcal{M}_P$ determines $U_1=U$, $U_2=U^\dagger e^{i\alpha} P$ and $U_{\tilde\gamma}=\mathrm{diag}(e^{i\phi_1},\hdots,e^{i\phi_k})$. 

If matrix $U_{\tilde\gamma}$ has a $d$-fold degeneracy in its spectrum, matrix $U_1U_2$ must be of the form $e^{i\alpha}PB$ where $B$ is a block-diagonal matrix with a $d\times d$ block forming a $U(d)$ matrix and ones outside the $d\times d$ block. Matrix $P$ is a permutation matrix from the quotient $S_k/S_d$.
Thus, we have components
\[\mathcal{M}_P^{(d)}\cong U(k)\times U(d)\times U(1)^{2}\times \frac{(U(1)^{k-d}-\Delta)}{S_{k-d}}.\] 
A tuple $(U,B,e^{i\alpha},e^{i\phi},[(e^{i\phi_1},\hdots,e^{i\phi_{k-d}})])\in \mathcal{M}_P^{(d)}$ determines $U_1=U$, $U_2=U^\dagger e^{i\alpha} PB$ and $U_{\tilde\gamma}=\mathrm{diag}(e^{i\phi},\hdots,e^{i\phi},e^{i\phi_1},\hdots,e^{i\phi_{k-d}})$. In the extreme case where $U_{\tilde\gamma}=e^{i\phi}\bone$, $U_1U_2$ can be any $U(k)$ matrix. Hence, in this case we have only one component
\[\mathcal{M}_0\cong U(d)\times U(d)\times U(1),\] 
where a tuple $(U,U',e^{i\phi})$, determines $U_1=U,\ U_2=U'$ and $U_{\tilde\gamma}=\mathrm{diag}(e^{i\phi},\hdots,e^{i\phi})$.

Summing up, we have obtained the following decomposition into connected components
\begin{gather}\label{theta-moduli}
\mathcal{M}_{4}(\Gamma_\Theta,U(k))=\mathcal{M}_0\sqcup\bigsqcup_{P\in S_k}\mathcal{M}_P\,\sqcup \\ \nonumber \sqcup\,\bigcup_{d=2}^{k-1}\bigsqcup_{P\in S_k/S_d}\mathcal{M}_P^{(d)}.
\end{gather}
The above decomposition simplifies slightly when specified to $k=2$. Then, the spectrum of $U_{\tilde\gamma}$ is either non-degenerate or $U_{\tilde\gamma}=e^{i\phi}\bone$. Component $\mathcal{M}_0=U(2)^{2}\times U(1)$. There are two "non-degenerate" components $\mathcal{M}_P$ that correspond to the identity and the transposition element of $S_2$. Both $\mathcal{M}_P$ are of the form $U(2)\times U(1)\times C_2(U_1)$ where $C_2(U_1)$ is a two-point unordered configuration space of $U(1)$ which is a topological circle. It is known that $C_2(S^1)$ is topologically $S^1$.
\end{example}

\section{Locally abelian anyons}
Following the concept of {\it generalised fractional statistics} on a torus which was introduced in \cite{Einarsson}, we show how to define analogous statistics on graphs using physical presentations of graph braid groups from section \ref{sec:exchanges}. The idea is to construct $U(k)$-representations of $Br_N(\Gamma)$ where to generating $Y$-loops we assign matrices of the form $e^{i\phi}\bone$ and only to generating $\mc{O}$-loops we assign general unitary matrices. The interpretation is that $Y$-loops correspond to exchanges of pairs of particles which are local in the sense that they are localised on junctions of $\Gamma$. Hence, $Y$-loops only utilise the local structure of $\Gamma$ as a star graph. On the other hand, $\mc{O}$-loops are global entities in the sense that they take a particle around a simple cycle in $\Gamma$ which can cross many junctions and hence they utilise the global structure of $\Gamma$. We say that anyons arising as such representations of graph braid groups are {\it locally abelian} anyons. This is because matrices from local $Y$-loops commute with each other and result with the multiplication of the multi-component wave function by an abelian phase factor.

 It has been shown in \cite{Einarsson} that quasiholes in certain Laughlin wave functions with periodic boundary conditions can be subject to generalised fractional statistics. Finding a physical model for locally abelian anyons on graphs is an open problem.
 
 Let us next show how locally abelian anyons are realised on a $\Theta$-graph.
 \begin{example}[Locally abelian anyons on a $\Theta$-graph]
 Let us examine the physical presentation of $Br_4(\Gamma_\Theta)$ that we derived in Example \ref{ex:theta-physical}. For the $Y$-loops we assign $\gamma\mapsto U_\gamma=e^{i\phi}\bone$, $\gamma'\mapsto U_{\gamma'}=e^{i\phi'}\bone$,\ $\gamma''\mapsto U_{\gamma''}=e^{i\phi''}\bone$. To the $\mc{O}$-loops we assign general $U(k)$ matrices $\alpha_{\mc{U}}\mapsto U_{\mc{U}}$ and $\alpha_{\mc{D}}\mapsto U_{\mc{D}}$. Relations between $\phi,\ \phi'$ and $\phi''$ can be derived from relators $R_1,\ R_2$ and $R_3$ in (\ref{theta-phys-N4}). In particular, because unitary matrices assigned to $Y$-loops are proportional to identity, they are invariant under conjugation. Hence, $R_1=\bone$ is satisfied automatically while $R_2=\bone$ and $R_3=\bone$ yield 
 \[\phi=\phi'=\phi''\ \mathrm{mod}\ 2\pi.\]
 Hence, locally abelian anyons from $Br_4(\Gamma_\Theta)$ are determined by an arbitrary choice of local exchange phase $\phi\in[0,2\pi[$ and global gauge $U(k)$ operators $U_{\mc{D}}$, $U_{\mc{U}}$. This exactly corresponds to component $\mathcal{M}_0$ of $\mathcal{M}_{4}(\Gamma_\Theta,U(k))$ in (\ref{theta-moduli}).
 \end{example}
 
\section*{Acknowledgements}
 TM gratefully acknowledges the financial support of the National Science Centre of Poland -- grants {\it Etiuda} no. $2017/24/T/ST1/00489$ and {\it Preludium} no. $2016/23/N/ST1/03209$.


\begin{thebibliography}{99}
\bibitem{ThoulessWu85} Thouless, D. J., Yong-Shi Wu, {\it Remarks on fractional statistics} Phys. Rev. B 31, 1191(R), 1985
\bibitem{Laughlin} Laughlin, R., B., {\it Anomalous Quantum Hall Effect: An Incompressible Quantum Fluid with Fractionally Charged Excitations}, Phys. Rev. Lett. 50, 1395, 1983
\bibitem{Haldane} Haldane, F., D., M., Rezayi, E., H., {\it Periodic Laughlin-Jastrow wave functions for the fractional quantized Hall effect}, Phys. Rev. B 31, 2529(R), 1985
\bibitem{Sudarshan} Imbo, T. D., Imbo C., S., Sudarshan, E., C., G., {\it Identical particles, exotic statistics and braid groups}, Phys. Lett. B, Vol. 234, I. 1-2, pp 103-107, 1990
\bibitem{LM} Leinaas, J., M., Myrheim, J., {\it On the theory of identical particles}, Nuovo Cim. 37B, 1-23, 1977
\bibitem{Souriau} Souriau, J. M., {\it Structure des systmes dynamiques}, Dunod, Paris, 1970
\bibitem{Wilczek} Wilczek, F., {\it Fractional statistics and anyon superconductivity},  Singapore: World Scientific, 1990
\bibitem{HKR} Harrison, J.,M., Keating, J., P., Robbins, J., M., {\it Quantum statistics on graphs} Proc. R. Soc. A vol. 467 no. 2125 212-23, 2011
\bibitem{alicea} Alicea, J., Oreg, Y., Refael, G., von Oppen, F., Fisher, M. P. A., {\it Non-Abelian statistics and topological quantum information processing in 1D wire networks}, Nature Physics 7, pp 412-417, 2011
\bibitem{Bolte13a} Bolte J., Kerner J., {\it Quantum graphs with singular two-particle interactions}, J. Phys. A: Math. Theor. 46 045206, 2013
\bibitem{HKRS} Harrison, J., M., Keating, J., P., Robbins, J., M., Sawicki A., {\it n-Particle Quantum Statistics on Graphs}, Comm. Math. Phys., Vol. 330, Issue 3, pp 1293-1326, 2014
\bibitem{kurlin} Kurlin, V., {\it Computing braid groups of graphs with applications to robot motion planning}, Homology Homotopy Appl.
Vol. 14, No. 1, pp 159-180, 2012
\bibitem{Knudsen} Byung Hee An, Drummond-Cole, G., C., Knudsen, B., {\it Subdivisional spaces and graph braid groups}, Documenta Mathematica 24 1-1000, 2019
\bibitem{ASphd} Sawicki, A., {\it Topology of graph configuration spaces and quantum statistics}, PhD thesis, Bristol, 2014
\bibitem{MS16} Maci\k{a}\.zek, T., Sawicki, A., {\it Homology groups for particles on one-connected graphs}, J. Math. Phys., Vol 58, no 6, 062103, 2017
\bibitem{MS18} Maci\k{a}\.zek, T., Sawicki, A., {\it Non-abelian quantum statistics on graphs}, arXiv:1806.02846, 2018
\bibitem{Einarsson} Einarsson, T., {\it Fractional statistics on a torus}, Phys. Rev. Lett. 64, 1995
\bibitem{Luck} L\"{u}ck, W., {\it Aspherical manifolds}, Bulletin of the Manifold Atlas 1-17, 2012
\bibitem{KoPark} Ko, K., H., Park, H., W., {\it Characteristics of graph braid groups}, Discrete \& Computational Geometry, Vol. 48, Issue 4, pp 915-963 (2012)
\bibitem{Forman} Forman, R., {\it Morse Theory for Cell Complexes}, Advances in Mathematics 134, 90145, 1998
\bibitem{FSbraid} Farley, D., Sabalka, L., {\it Discrete Morse theory and graph braid groups}, Algebr. Geom. Topol. 5 1075-1109, 2005
\bibitem{FS12} Farley, D., Sabalka, L., {\it Presentations of graph braid groups}, Forum Math. 24 827-859, 2012
\bibitem{FStree} Farley, D., Sabalka, L., {\it On the cohomology rings of tree braid groups}, J. Pure Appl. Algebra 212 53-71, 2008
\bibitem{AS12} Sawicki, A., {\it Discrete Morse functions for graph configuration spaces}, J. Phys. A: Math. Theor. 45 505202, 2012
\bibitem{AbramsPhD} Abrams, A., {\it Configuration spaces and braid groups of graphs}, Ph.D. thesis, UC Berkley, 2000
\bibitem{Abrams-subdiv} Prue, P., Scrimshaw, T., {\it Abrams's stable equivalence for graph braid groups}, Topology and its Applications, Vol. 178, pp 136-145, 2014
\bibitem{swiatkowski} \'{S}wi\k{a}tkowski, J., {\it Estimates for homological dimension of configuration spaces of graphs}, Colloq. Math. 89(1), 69-79, https://doi.org/10.4064/cm89-1-5, 2001
\bibitem{BF09} Barnett, K., Farber, M., {\it Topology of Configuration Space of Two Particles on a Graph}, Algebraic \& Geometric Topology 9 pp 593-624 (2009)
\bibitem{code} Maci\k{a}\.zek, T., {\it An implementation of discrete Morse theory for graph configuration spaces}, \url{www.github.com/tmaciazek/graph-morse}, 2019
\end{thebibliography}
\end{document}